\documentclass{article}
\usepackage{ijcai23}

\usepackage{times}
\usepackage{soul}
\usepackage{url}
\usepackage[hidelinks]{hyperref}
\usepackage[utf8]{inputenc}
\usepackage[small]{caption}
\usepackage{graphicx}
\usepackage{mathrsfs}
\usepackage{amsfonts}
\usepackage{amssymb}
\usepackage{amsmath}
\usepackage{amsthm}
\usepackage{booktabs}
\usepackage{algorithm}
\usepackage{algorithmic}
\usepackage[switch]{lineno}

\usepackage{balance} 
\usepackage{comment}
\usepackage{thm-restate}

\usepackage{tikz}
\usepackage{subcaption}
\usepackage{enumitem}
\usepackage{ifthen}
\usepackage{multirow,array}
\usepackage{pifont}
\usepackage{xcolor}
\definecolor{d_green}{RGB}{0,200,0}
\definecolor{d_red}{RGB}{200,0,00}
\newcommand{\cmark}{\textcolor{d_green}{\ding{51}}}%
\newcommand{\xmark}{\textcolor{d_red}{\ding{55}}}%
\usetikzlibrary{arrows,automata,positioning,calc}

\newtheorem{example}{Example}
\newtheorem{theorem}{Theorem}
\newtheorem{proposition}{Proposition}
\newtheorem{corollary}{Corollary}

\def\SIGMA#1{\Sigma^{p}_{#1}}
\def\PI#1{\Pi^{p}_{#1}}

\newcommand{\putaway}[1]{}

\def\prefeq{\succeq}
\def\pref{\succ}

\def\_{\mbox{-}}
\def\~{\mbox{$\sim$}}

\def\set#1{\{#1\}}		
\def\to{\rightarrow}		
\let\.=\cdot

\def\false{\bot}

\def\Bool{\mathbb{B}}               
\def\Rat{\mathbb{Q}}                
\def\phi{\varphi}

\def\calG{\mathcal{G}}  
 \def\calK{\mathcal{K}} 
  \def\calO{\mathcal{O}}

\def\calV{\mathcal{V}}  
 
\usepackage{tikz}
\usepackage{subcaption}
\usepackage{booktabs}
\usepackage{enumitem}
\usepackage{ifthen}
\usepackage{comment}
\usepackage{bm}
\usepackage{fancyhdr}
\usetikzlibrary{arrows,automata,positioning,calc}

\newcommand{\acv}{\vec{v}}

\newcommand{\obs}{\vec{o}}
\newcommand{\obsset}{\Bool^{\calO}}
\newcommand{\env}[1]{\textsc{env}(#1)}

\newcommand{\nasheq}{\textsc{ne}}
\newcommand{\consistent}[1]{\textsc{cons}(#1,\obs)}
\newcommand{\ind}[1]{\textsc{ind}(#1)}
\newcommand{\init}[1]{\textsc{init}(#1)}
\newcommand{\soft}[1]{\textsc{soft}(#1)}
\newcommand{\hard}[1]{\textsc{hard}(#1)}
\newcommand{\obseq}{=_{\calO}}
\newcommand{\obsneq}{\neq_{\calO}}
\newcommand{\resobs}{\vert_{\calO}}

\newcommand{\indev}[1]{\textsc{in}(#1,D)}
\newcommand{\resp}[1]{\mathsf{resp}(#1)}

\title{Principal-Agent Boolean Games}

\author{
David Hyland$^1$
\and
Julian Gutierrez$^2$
\And
Michael Wooldridge$^1$
\affiliations
$^1$University of Oxford, UK\\
$^2$Monash University, Australia\\
\emails
\{david.hyland, mjw\}@cs.ox.ac.uk,
julian.gutierrez@monash.edu
}

\begin{document}

\maketitle

\begin{abstract}
    We introduce and study a computational version of the principal-agent problem -- a classic problem in Economics that arises when a principal desires to contract an agent to carry out some task, but has incomplete information about the agent or their subsequent actions. The key challenge in this setting is for the principal to design a contract for the agent such that the agent's preferences are then aligned with those of the principal. We study this problem using a variation of Boolean games, where multiple players each choose valuations for Boolean variables under their control, seeking the satisfaction of a personal goal, given as a Boolean logic formula. In our setting, the principal can only observe some subset of these variables, and the principal chooses a contract which rewards players on the basis of the assignments they make for the variables that are observable to the principal. The principal's challenge is to design a contract so that, firstly, the principal's goal is achieved in some or all Nash equilibrium choices, and secondly, that the principal is able to verify that their goal is satisfied. In this paper, we formally define this problem and completely characterise the computational complexity of the most relevant decision problems associated with it.
\end{abstract}

\section{Introduction}
The principal-agent problem is a classic problem in economics. It arises when a principal aims to sub-contract a task to an agent (or agents) while having only imperfect information about the agent(s). The two key classes of principal-agent problems are \textit{adverse selection} and \textit{moral hazard}. In a setting of adverse selection, the principal may have partial information about the agents' types, which may be relevant to their ability to carry out the delegated task. For example, when an employer hires a new employee, they have only partial information about the skills and aptitude of the employee. Their skills and aptitude will then play a part in the performance of the delegated task. In moral hazard settings, upon which we focus in the present paper, the principal is unable to observe the actions taken by agents subsequent to the contract. In particular, the principal may be unable to directly verify that the terms of the contract have been respected. For example, suppose we want to hire a cleaner to ensure that a particular building is kept clean, although we are only able to inspect the building at infrequent intervals: if the cleaner desires to minimise effort, why would they then not simply clean up only when an inspection is imminent? The key challenge here is to design a contract for the agent that aligns the incentives of the agent with those of the manager, so that the manager can have confidence that the agent will be rationally motivated to expedite the task at hand, even if the manager is unable to directly verify that this is indeed the case. 

Principal-agent problems of this type are extremely relevant in Artificial Intelligence as well as in Multi-Agent Systems: the venerable Contract Net protocol, for example, deals with precisely the situation where a principal delegates tasks to an agent~\cite{smith:80a}, and although they were not considered in the original work, informational asymmetries between the principal and the agent may of course occur. 

Our present work investigates a computational version of the principal-agent problem based on the Boolean games framework~\cite{harrenstein:2001a}. Boolean games are non-cooperative games in which players each control a finite set of Boolean variables, and the pure strategies available to a player correspond to the set of possible assignments of truth or falsity to those variables. Preferences are captured by assigning each player a propositional logic formula that the player desires to see satisfied. In the variant of Boolean games that we adapt in the present paper, secondary preferences are defined by costs associated with assignments of truth or falsity -- a player first seeks to satisfy their goal, and secondarily seeks to minimise costs~\cite{wooldridge2013incentive}.

Classical work in the area of moral hazard in teams fails to take into account the potential lexicographic or qualitative nature of agent preferences. To our knowledge, our work is the first to study multi-agent moral hazard problems in a game where agents have lexicographic qualitative and quantitative preferences. In our model, a principal has a goal they desire to see accomplished, expressed, as with the agents' goals, as a propositional formula. The principal chooses a contract that rewards individual agents on the basis of their observable variables. The fundamental question we ask is whether it is possible for the principal to design a contract such that, if agents rationally make choices (i.e., make choices that form a game-theoretic equilibrium), then, firstly, the principal's goal will be accomplished, and secondly, the principal can formally and automatically \emph{verify} that the goal is accomplished, under the assumption that agents have acted rationally.

We derive complexity results for two classes of problems: Can the principal formally verify whether or not their objective, given as a propositional logic formula over the observable variables in the game, is satisfied on some/all observations and can the principal design a contract such that their objective is satisfied on some/all possible observations.

\section{Related Work} \label{sec:related_work}

\vspace*{1ex}\noindent\textbf{Moral Hazard:} Investigations of the moral hazard problem were initially motivated by the desire to better understand the relationship between risk and incentives in insurance contracts~\cite{pauly1968economics,arrow1971insurance}. Early models of moral hazard sought to study how risks can be shared optimally between principal and agent using contracts~\cite{michael1976theory,holmstrom1979moral,shavell1979risk,grossman1992analysis}.\footnote{See \cite{georgiadis2022contracting} for an overview of moral hazard studies.} Extensions of these models to settings with \textit{multiple agents} and a single principal were also developed, in which the presence of competition between agents must be taken into account~\cite{holmstrom1982moral,itoh1991incentives,che2001optimal}. More recently, computational approaches to contract design have extended the original models in various ways and taken an algorithmic lens on computing optimal contracts \cite{dutting2019simple,dutting2021complexity,alon2021contracts,guruganesh2021contracts,gan2022optimal,dutting2023ambiguous}. Multi-agent contracts have also been studied from this perspective, which introduces the challenge of reasoning about combinatorially large action spaces \cite{babaioff2006combinatorial,emek2012computing,dutting2022multi,castiglioni2023multi}.
Common among these studies is the assumption that agent preferences are quantitative, and partial observability is primarily modelled as a noisy signal of the agents' actions, which are often represented as a quantitative degree of effort. In contrast, our model allows agents to have qualitative goals specified as propositional formulae, which are prioritised over their quantitative objectives and constrains the agents' decision-making to a discrete action space. Furthermore, the principal's objective in our model is a propositional logic formula, as opposed to a function of the observation signal. An example of a setting where this may be more appropriate is a scenario in which a set of tasks are divided among a group of autonomous agents, each of whom are capable of reliably completing their assigned tasks. Each task is associated with a cost of completion, and the agents must decide which tasks to complete based on their objectives and the contract payments that are offered under the limited principal's observational abilities.

\vspace*{1ex}\noindent\textbf{Mechanism Design:} There are computational approaches to mechanism design in which approximation algorithms~\cite{nisan1999algorithmic}, program synthesis~\cite{narayanaswamiautomating}, and satisfiability checking~\cite{maubert2021strategic,mittelmann2022automated}, among others, have been proposed to automatically construct optimal mechanisms. These works assume a stronger degree of control by the principal over the agents' interactions and their utilities.

\vspace*{1ex}\noindent\textbf{Lexicographic Qualitative and Quantitative Preferences:}
Finally, there have been several studies devoted to understanding solution concepts in games with lexicographic qualitative and quantitative preferences~\cite{chatterjee2005mean,bloem2009better,gutierrez2017nash,gutierrez2021equilibria,bulling2022combining}, but these do not consider problems of incentive design. Other previous work has studied the idea of a principal introducing a taxation scheme into games where players have a very similar preference structure to the one we consider~\cite{wooldridge2013incentive,harrenstein:2014a,harrenstein:2017a,levit2019incentive}. 
Additionally,~\cite{gutierrez_et_al:LIPIcs:2019:10924} study a related problem called `equilibrium design', but they consider players with only quantitative objectives. A key assumption in such previous works was that the principal can \textit{fully observe} the outcome $\acv$ selected by the agents. However, the principal does not have this luxury in many real-world situations, and it is this issue that motivates the study of the principal-agent problem. Other studies of incomplete information in Boolean games introduce partial observability at the level of the agents, as opposed to an external principal who seeks to influence their behaviour \cite{agotnes2013verifiable,clercq2014possibilistic}.

\section{Preliminaries} \label{sec:prelim}
\vspace*{1ex}\noindent\textbf{Notation:}
 Where $S$ is a set, we denote the powerset of $S$ by $\mathbf{2}^S$. We make use of propositional logic over Boolean variables throughout this study. In this context, we will let $\Bool = \set{\top,\false}$ denote the set of Boolean literals. Where $S$ is a set of Boolean variables, we let $\Bool^S$ denote the set of possible \textit{valuations}, which are combinations of truth assignments to each variable in $S$. Where $T$ is also a set of Boolean variables and $v \in \Bool^S$ is a valuation in $S$, we will let $v\vert_{T}$ denote the \textit{restriction} of $v$ to $T$, which is the Boolean valuation $v' \in \Bool^{S\cap T}$ that agrees with $v$ on all variables in $S \cap T$.
 
\vspace*{1ex}\noindent\textbf{Boolean Games:}
In this study, we will model a multi-agent moral hazard problem in the context of one-shot Boolean Games with costs, as presented in \cite{wooldridge2013incentive}. A \textit{Boolean Game with costs} (henceforth ``Boolean game'' or simply ``game'') is a structure \[ B = (N,\Phi,(\Phi_i)_{i \in N}, (\gamma_i)_{i \in N}, (c_i)_{i \in N}), \] where: $N = \set{1,\ldots,n}$ is a set of \textit{agents}; $\Phi =~\set{p_1,p_2,\ldots,p_m}$ is a finite set of $m \geq n$ \textit{Boolean variables}; For each $i \in N$, the set $\Phi_i$ is the non-empty set of \textit{Boolean variables uniquely controlled by agent} $i$, such that the collection $(\Phi_i)_{i\in N}$ forms a partition of $\Phi$; For each  $i \in N$, formula $\gamma_i$ is the \textit{goal} of agent $i$, which is represented as a propositional formula over $\Phi$; For each $i \in N$, with $c_i: \Bool^{\Phi} \to \Rat_+$ we represent the \textit{cost function} of $i$, which assigns to each valuation $\acv$ of the variables in $\Phi$ a non-negative cost $c_i(\acv)$, indicating the cost incurred by $i$ under the valuation. We let $c_i^*$ denote the maximal cost that  $i$ can incur under their $c_i$.

A \textit{strategy} $v_i : \Phi_i \to \set{\top,\bot}$ for agent $i$ is defined as an assignment of truth values to every variable in $\Phi_i$. Here, we will use $1$ and $0$ interchangeably with $\top$ and $\bot$, respectively, in the context of strategies. For each agent $i \in N$, let $V_i$ be their set of possible strategies. A \textit{strategy profile} or \textit{outcome} is then a tuple of strategies $\acv = (v_1,\ldots,v_n)$, where agents select their strategies simultaneously in the absence of communication and knowledge of what strategies other agents selected. Additionally, given an agent $i \in N$, a strategy profile $\acv$, and an individual strategy $v_i' \in \Bool^{\Phi_i}$ for $i$, we let $(\acv_{-i},v_i')$ be the strategy profile obtained by replacing $v_i$ in $\acv$ with $v_i'$, i.e., the strategy profile $(v_1,\ldots,v_{i-1},v_i',v_{i+1},\ldots,v_n)$. Given a game $B$, we also define its \textit{costless} version as the game $B^0 = (N,\Phi,(\Phi_i)_{i \in N}, (\gamma_i)_{i \in N}, (c_i^0)_{i \in N})$ which is as game $B$, except that $c_i^0(\acv) = 0$ for all $\acv \in \Bool^{\Phi}$ and $i \in N$.


For a strategy profile $\acv$ and a propositional logic formula $\phi$ over the set of Boolean variables $\Phi$, we write $\acv \models \phi$ to mean that $\acv$ satisfies $\phi$, where $\models$ denotes the propositional satisfaction relation. If it holds that $\acv \models \gamma_i$ for some agent $i \in N$, we say that agent $i$'s goal is satisfied under $\acv$ and $i$ is a \textit{winner} under $\acv$. Any agent whose goal is not satisfied under a strategy profile $\acv$ will be referred to as a \textit{loser} under $\acv$. We write $W(\acv)$ and $L(\acv)$ to denote the sets of winners and losers under a given strategy profile $\acv$, respectively.

\vspace*{1ex}\noindent\textbf{Utilities, Preferences, and Nash Equilibrium:}
In order to model both qualitative \textit{and} quantitative preferences, we consider a model where agent preferences are defined according to a lexicographic ordering on goal satisfaction and their incurred costs. Specifically, an agent $i$ prefers any outcome in which their goal $\gamma_i$ is satisfied over any outcome in which it is not, no matter what cost they incur. Secondarily, each agent prefers to minimise the cost that they incur. Formally, we define the \textit{utility function} $u_i$ for each agent $i \in N$ given a strategy profile $\acv$ as follows:
\[ u_i(\acv) = 
\left\{\begin{array}{ll}
1+c_i^*-c_i(\acv)   & \mbox{if $\acv \models \gamma_i$}\\
-c_i(\acv)          & \mbox{otherwise.}
\end{array}\right.
\]
Observe that if an agent~$i$ achieves their goal, then their utility will lie in the range~$[1,c_i^* + 1]$; if their goal is not achieved, then their utility will lie in the range~$[-c_i^*,0]$. Preference relations~$\prefeq_i$ over outcomes for each~$i \in N$ are defined in the obvious way: $\acv_1 \prefeq_i \acv_2$ if and only if $u_i(\acv_1) \geq u_i(\acv_2)$, with indifference relations~$\sim_i$ and strict preference relations~$\pref_i$ defined in the usual way.

The primary solution concept we will work with is the pure strategy Nash equilibrium. Formally, a strategy profile $\acv$ is a \textit{Nash equilibrium} of a Boolean game $B$ with cost function $c$ if there is no agent $i \in N$ and strategy $v_i' \in V_i$ such that $(\acv_{-i},v_i') \pref_i \acv$. If such a strategy exists for an agent $i \in N$, we say that $i$ has a \textit{beneficial deviation} from $\acv$ to $(\acv_{-i},v_i')$. Where $B$ is a game, we write~$\nasheq(B)$ to denote the set of Nash equilibria of~$B$. Additionally, if~$\phi$ is a Boolean logic formula over $\Phi$, we let $\nasheq_{\phi}(B) = \set{\acv \in \nasheq(B)\ \mid\ \acv \models \phi}$ denote the set of Nash equilibria of $B$ which satisfy formula~$\phi$.

\section{The Principal-Agent Verification Problem}

In this study, we aim to formulate and analyse the multi-agent moral hazard problem in the context of Boolean games with costs, where hidden actions are modelled by the principal only being able to observe the values of some subset $\calO \subseteq \Phi$, known as the \textit{observable set}. An \textit{observation} by the principal is defined to be an assignment of truth values to all variables in its observable set $\calO$, and is denoted by $\obs$.\footnote{A special case of this is when we have $\calO = \emptyset$, where the principal cannot observe the values of any of the Boolean variables in the game. In this case, we can say that the principal makes a \textit{null observation}, which is written as $\obs = ()$.} In general, a single observation $\obs$ may correspond to more than one underlying strategy profile $\acv$ if not all of the variables are fully observable. Thus, the observable set $\calO$ naturally gives rise to the possibility for strategy profiles to be \textit{indistinguishable} from one another. This can be expressed formally as an observational equivalence relation $\obseq$, defined over the set of observations $\obsset$ such that for all strategy profiles $\acv, \acv' \in V$, we say that $\acv$ is \textit{indistinguishable} from $\acv'$, written $\acv \obseq \acv'$, if it is the case that $\acv\resobs = \acv'\resobs$. If it is \textit{not} the case that $\acv \obseq \acv'$, then we say that strategy profile $\acv$ is \textit{distinguishable} from strategy profile $\acv$ and write $\acv \obsneq \acv'$ in that case. 


The problem faced by the principal is to design a contract so that they are able to ensure/verify that their objective has been accomplished on some or all of the Nash equilibria under the contract. To begin with, we will first analyse the simple case where the principal does not intervene, but simply asks whether they can \textit{verify} whether or not some or all Nash equilibria consistent with any observation will satisfy their goal, given as a propositional logic formula~$\phi$. The value of this problem is that it provides a baseline against which the principal can assess the benefit of any potential contract that is designed. In order to formally state these problems, we first make precise the notion of consistency. Formally, we say that a strategy profile $\acv$ is \textit{consistent} with an observation $\obs$ if it is the case that $\acv\resobs = \obs$. Then, we define the \textit{consistent set} of an observation $\obs$ in a given Boolean game $B$ in the following way: $
\consistent{B} = \set{\acv \in \nasheq(B)\ \mid\ \acv \mbox{ is consistent with } \obs}.$
Finally, we define the set $\env{B,\obs,\phi} = \nasheq_{\phi}(B) \cap \consistent{B}$ to be the set of Nash equilibria in game $B$ that satisfy $\phi$ and are consistent with $\obs$. The first decision problem of interest is then stated as follows:

\begin{quote}
\underline{\textsc{E-Nash Verifiability}}:\\
\emph{Given}: Game $B$, observation $\obs \in \obsset$, formula $\phi$.\\
\emph{Question}: Is it the case that $ \env{B,\obs,\phi} \neq \emptyset$?
\end{quote}

Whilst an answer to the E-Nash Verifiability problem provides useful information to the principal about the \textit{possibility} of their goal being achieved on some Nash equilibrium that is consistent with an observation, it does not make any guarantees as to whether or not the principal's goal has been achieved, given what they have observed. To obtain such a guarantee, we require the natural counterpart to this problem -- the \textit{A-Nash Verifiability} problem\footnote{Notation E- and A- are used to indicate existential and universal reasoning about the set of Nash equilibria in a game respectively.}:

\begin{quote}
\underline{\textsc{A-Nash Verifiability}}:\\
\emph{Given}: Game $B$, observation $\obs \in \obsset$, formula $\phi$.\\
\emph{Question}: Is the case that $\consistent{B} \subseteq \nasheq_{\phi}(B)$?
\end{quote}

If the answer to the A-Nash Verifiability problem is ``yes'', then the principal can be assured that any Nash equilibrium consistent with their observation satisfies their goal. The following complexity results of the E- and A-Nash Verifiability problems can be readily derived from existing results relating to Boolean games~\cite{wooldridge2013incentive}:

\begin{proposition}\label{proposition:nash-verifiability}
\textsc{E-Nash Verifiability} is $\SIGMA2$-complete, and \textsc{A-Nash Verifiability} is $\PI2$-complete.
\end{proposition}


\begin{example}
To illustrate the differences between the E-Nash and the A-Nash verifiability problems, consider the game $B_1$ consisting of two players with $\Phi_1 = \set{p_1}$, $\Phi_2 = \set{p_2}$, goals $\gamma_1 = \top$, $\gamma_2 = \neg p_1$, cost function as given by the table on the left in Figure \ref{fig:example_games}, and observable set $\calO = \set{p_1}$. Now suppose that the principal's objective is $\phi = p_2$ and consider the E- and A-Nash verifiability problems for each observation $\obs \in \obsset = \set{0,1}$. Firstly, we can identify the Nash equilibria in $B_1$ as being $(0,0)$ and $(0,1)$, with only the latter satisfying $\phi$. Under the observation $\obs_1 = (1)$, there is no Nash equilibrium in $B_1$ that is consistent with $\obs_1$, so the answer to E-Nash verifiability is ``no'' and the answer to A-Nash verifiability is ``yes'' for $\obs_1$. However, under the observation $\obs_2 = (0)$, it is the case that $\consistent{B} \cap \nasheq_{\phi}(B) \neq \emptyset$ but not the case that $\consistent{B} \subseteq \nasheq_{\phi}(B)$. Thus, the answer to E-Nash verifiability for $\obs_2$ is ``yes'' while it is ``no'' for A-Nash verifiability. 
\end{example}

\begin{figure}
    \centering
    \begin{tabular}{cc|c|c|}
      & \multicolumn{1}{c}{} & \multicolumn{1}{c}{$1$}  & \multicolumn{1}{c}{$0$} \\\cline{3-4}
      & $1$ & $(2,0)$ & $(3,0)$ \\\cline{3-4}
      & $0$ & $(1,1)$ & $(1,1)$ \\\cline{3-4}
    \end{tabular}
    \begin{tabular}{cc | c | c |}
      & \multicolumn{1}{c}{} & \multicolumn{1}{c}{$1$}  & \multicolumn{1}{c}{$0$} \\ \cline{3-4}
      & $1$ & $(10,1)$ & $(10,2)$ \\ \cline{3-4}
      & $0$ & $(1,2)$ & $(1,1)$ \\ \cline{3-4}
    \end{tabular}
    \caption{Two cost matrices representing the agents' costs incurred under different strategy profiles in two different Boolean games. The game represented by the left figure is referred to as $B_1$ and discussed in Example 1. The game represented by the right figure is referred to as $B_2$ and discussed in Example 2. Player $1$ (the row player) controls variable $p_1$ and player $2$ (the column player) controls variable $p_2$.}
    \label{fig:example_games}
\end{figure}

\section{Contract Design}

The verification problem studied thus far allows the principal to obtain helpful information about whether or not they can be confident that their goal was satisfied under any possible observation that is consistent with at least one Nash equilibrium. However, in the classical formulations of the moral hazard problem, a principal is tasked with \textit{designing} a contract to align the agents' incentives with their own. In this study, the principal's limitation lies not in an inability to accurately observe the outcome of its observables, but rather in their ability to only observe some subset of the actions chosen, as well as the agents' lack of willingness to compromise on satisfying their goal, no matter what the incentives are. For this, we first introduce the concept of contracts, which amount to an alteration of the costs incurred by the agents for taking various actions. This is similar to the notion of taxation schemes introduced in \cite{wooldridge2013incentive}, but with the critical distinction that in our setting, the principal can only reward agents based on the outcomes of \textit{observable} variables. 

Formally, a contract $\kappa = (\kappa_1,\ldots,\kappa_n)$ is a tuple of functions~$\kappa_i : \obsset \to \Rat_{\geq 0}$ for each agent $i \in N$ that map each possible observation by the principal to a non-negative rational number. Let $\calK$ denote the set of all contracts for a given game and let $\kappa_i^*$ denote the highest possible payment that the principal offers to each agent $i \in N$ over all possible observations. Then, given a game $B$ and an observation set $\calO$, a contract $\kappa$ gives rise to a new utility function $u_i^\kappa$ for each agent $i \in N$ given a strategy profile $\acv$:
\[
u_i^{\kappa}(\acv) = \left\{\begin{array}{ll}
1 + c_i^* + \kappa_i(\acv\resobs) - c_i(\acv)   & \mbox{if $\acv \models \gamma_i$}\\
-\kappa_i^* + \kappa_i(\acv\resobs) - c_i(\acv)          & \mbox{otherwise.}
\end{array}\right.
\]
Defining the utility of agents under a given contract in this way captures two desirable properties. Firstly, it preserves the lexicographic preferences of agents -- each agent is guaranteed to achieve a strictly positive utility if their goal is satisfied, whereas they can achieve a utility of at most zero if it is not. Secondly, regardless of whether an agent's goal is satisfied, their utility is strictly increasing in the payment received from the principal and strictly decreasing in the costs they incur. Given a game $B$, an observable set $\calO$ and a contract $\kappa$, we define the \textit{Boolean game induced by} $\kappa$ to be the game $B^{\kappa}$ which is as game $B$, but where each agent $i$'s utility is given by $u_i^{\kappa}$. Given a contract $\kappa$, an agent $i \in N$ and two strategy profiles $\acv^1,\acv^2$, we write $\acv^1 \prefeq_i^{\kappa} \acv^2$ if and only if $u_i^{\kappa}(\acv^1) \geq u_i^{\kappa}(\acv^2)$ and define $\pref_i^{\kappa}$ in the obvious manner.


\subsection*{Inducing and Eliminating Equilibria}

Due to the partial observability of the agents' actions and the fact that agents will always prefer to achieve their goal than not to do so, there are limits to the ability of a principal to either introduce or eliminate equilibria from a given game. As a special case, note that if all agents' actions are fully observable ($\calO = \Phi$), then the contract problem can be characterised by the game's hard and soft equilibria -- respectively those that can not be eliminated under any contract and those that can~\cite{wooldridge2013incentive,harrenstein:2014a,harrenstein:2017a}. Thus, we will focus on the case where $\calO \subset \Phi$. In this scenario, it is not in general possible to characterise the manipulability of a given game~$B$ by analysing the qualitative structure of the underlying costless Boolean game, as the principal is no longer able to eliminate the cost considerations for all agents by designing a contract so that each agent's quantitative payoff is constant regardless of their actions. Thus, the ability of the principal to induce or eliminate equilibria by means of contract design will be critically dependent on the initial cost functions of the agents. 

\begin{example}
As an example of the importance of initial costs in this model, consider a game $B_2$ consisting of two players with goals $\gamma_1 = \gamma_2 = \top$, cost function as given by the table on the right in Figure \ref{fig:example_games}, and observable set $\calO = \set{p_1}$. Now suppose again that the principal's objective is $\phi = p_2$. With $\bar{p}_x$, $x\in\{1,2\}$, we denote $\neg p_x$. Firstly, note that the only Nash equilibrium outcome of $B_2$ is $(0,0)$. Now, because the principal cannot observe $p_2$ directly, the only way for them to ensure that $p_2$ is rationally chosen by player 2 is to offer a contract to \textit{player $1$} to make the choice of setting $p_1 = 1$ more rational than setting it to $0$. However, in order to do so, such a contract $\kappa$ must satisfy $\kappa_1(p_1) > \kappa_1(\bar{p}_1) + 9$, under which the new unique Nash equilibrium becomes $(1,1)$. This illustrates a scenario in which it is lucrative for an employee to control what is observable by the employer so as to benefit from indirect incentivisation.
\end{example}
    
Before proceeding, we find it useful to define some terminology~\cite{harrenstein:2014a}. Firstly, we define the set $\init{B} = \nasheq(B^0)$ of \textit{initial equilibria} to be the set of Nash equilibria in the costless game $B^0$. Formally, we will say that given a game $B$ and an observable set $\calO$, a strategy profile $\acv$ is an \textit{inducible equilibrium} if there exists a contract $\kappa$ such that $\acv \in \nasheq(B^{\kappa})$, and we let $\ind{B,\calO}$ be the set of inducible equilibria of a game $B$ with respect to an observable set $\calO$. Additionally, we will say that a strategy profile $\acv$ is a \textit{hard equilibrium} (with respect to $\calO$), written $\acv \in \hard{B,\calO}$, if $\acv \in \nasheq(B)$ and there is no contract $\kappa$ such that $\acv \notin \nasheq(B^{\kappa})$. Complementarily, we say that a strategy profile $\acv$ is a \textit{soft equilibrium}, written $\acv \in \soft{B,\calO}$, if $\acv \in \nasheq(B) \setminus \hard{B,\calO}$. It is easy to verify that the previously defined sets are related in the following way:

\begin{proposition}\label{proposition:set_relations}
For all Boolean games $B$ and observable sets $\calO$, it holds that $\hard{B,\calO} \subseteq \ind{B,\calO}$, $\soft{B,\calO} \subseteq \ind{B,\calO}$, and $\nasheq(B^{\kappa}) \subseteq \ind{B,\calO} \subseteq \init{B}$ for all contracts $\kappa$.
\end{proposition}

To reason about the principal's ability to design effective contracts, a method of analysing the inducibility and eliminability of sets of strategy profiles is needed. Following the approach of \cite{harrenstein:2017a}, we will introduce different types of potential deviations between strategy profiles. Suppose that $\acv, \acv' \in V$ are two distinct strategy profiles such that $\acv' = (\acv_{-i},v_i')$ for some $i \in N$ and $v_i' \in V_i \setminus \set{v_i}$. Then we say that:
\begin{itemize}
    \item $\acv'$ is an \textit{initial deviation} for $i$ from $\acv$, written $\acv \rightharpoonup_i \acv'$, if $i \in W(\acv) \Rightarrow i \in W(\acv')$.
    \item $\acv'$ is an \textit{inducible deviation} for $i$ from $\acv$, written $\acv \to_i \acv'$, if there is a contract $\kappa$ such that $\acv' \pref_i^{\kappa} \acv$. Under such a contract, we say that $\kappa$ \textit{induces} a deviation from $\acv$ to $\acv'$.
    \item $\acv'$ is a \textit{soft deviation} for $i$ from $\acv$, written $\acv \leftrightarrows_i \acv'$, if both $ \acv \to_i \acv'$ and $ \acv' \to_i \acv$. 
    \item $\acv'$ is a \textit{hard deviation} for $i$ from $\acv$, written $\acv \rightrightarrows_i \acv'$, if $ \acv \to_i \acv'$ but not $\acv' \to_i \acv$.
    
\end{itemize}

Note that just because a strategy profile is an initial deviation from another strategy profile, this does not necessarily imply that a contract exists which could \textit{induce} that deviation, i.e., convert it into a beneficial deviation. Instead, we need the stronger notion of \textit{inducible deviations}, which is, in turn, used to define soft and hard deviations. The following proposition characterises the relationship between initial and inducible deviations. The proofs of the following three Propositions proceed via a straightforward case analysis using the definitions of the previously defined concepts, so we omit them for the sake of space. For the backward implications, contracts can be designed to offer $c_i^*+1$ to an agent $i$ in order to make a particular observation more appealing than another observation, under which no reward is offered.

\begin{proposition}\label{proposition:induce_dev}
Let $B$ be a game, $\calO \subseteq \Phi$ an observable set, and $\acv,\acv' \in V$ be two distinct strategy profiles such that $\acv' = (\acv_{-i},v_i')$ for some $i \in N$ and $v_i' \in V_i \setminus \set{v_i}$. Then, $\acv \to_i \acv'$ if and only if $\acv \rightharpoonup_i \acv'$ and one of the following conditions hold: 1) $\acv \obsneq \acv'$, or 2) $\acv \obseq \acv'$ and $\acv' \pref_i \acv$.
\end{proposition}

This result highlights the fact that the ability to design contracts to induce deviations between two strategy profiles is wholly determined by the initial cost and goal structure of the game in the case where the strategy profiles are indistinguishable. Next, we present a characterisation of the set of inducible equilibria of a game, which tells a principal when it is possible to design a contract $\kappa$ such that a particular outcome becomes a Nash equilibrium under $\kappa$:

\begin{proposition}\label{proposition:inducible_eqm}
Let $B$ be a game, $\calO \subseteq \Phi$ an observable set, and $\acv$ a strategy profile. Then, $\acv \in \ind{B,\calO}$ if and only if the following conditions hold: 1) $\acv \in \init{B}$; and 2) For all agents $i \in N$, and all choices $v_i' \in V_i$ such that $(\acv_{-i},v_i') \obseq \acv$ and $i \in W(\acv) \Leftrightarrow i \in W(\acv_{-i},v_i')$, we have that $c_i(\acv_{-i},v_i') \geq c_i(\acv)$.
\end{proposition}

Now, we focus on characterising the sets of hard and soft equilibria, which allows a principal to check whether an undesirable initial equilibrium can be eliminated from a game.

\begin{proposition}\label{proposition:soft}
Let $B$ be a game, $\calO \subseteq \Phi$ an observable set, and $\acv$ a strategy profile. Then, $\acv \in \soft{B,\calO}$ if and only if the following conditions hold: 1) $\acv \in \init{B}$; and 2) There is agent $i \in N$ and a strategy $v_i' \in V_i \setminus \set{v_i}$ such that $(\acv_{-i},v_i') \obsneq \acv$ and $i \in W(\acv_{-i},v_i') \Leftrightarrow i \in W(\acv)$.
\end{proposition}

Because $\hard{B} = \nasheq(B) \setminus \soft{B}$ by definition, a characterisation of the set of hard equilibria in a game follows immediately from Proposition \ref{proposition:soft}:

\begin{corollary}\label{corollary:hard}
Let $B$ be a game, $\calO \subseteq \Phi$ a non-empty observable set, and $\acv$ a strategy profile. Then, $\acv \in \hard{B,\calO}$ if and only if the following conditions hold: 1) $\acv \in \nasheq(B)$; and 2) For all agents $i \in N$ and all strategies $v_i' \in V_i \setminus \set{v_i}$ such that $(\acv_{-i},v_i') \obsneq \acv$, we have $(\acv_{-i},v_i') \rightrightarrows_i \acv$.
\end{corollary}

So far, our analysis has focused on whether or not a single strategy profile is an inducible, hard, or soft equilibrium. However, it is not sufficient in general to consider individual strategy profiles in this way, as there may be several strategy profiles that the principal would like to induce or eliminate. The more general problem faced by the principal is thus how it can determine whether a contract can be designed to jointly eliminate several equilibria in a game.

Therefore, we will extend the concept of eliminability to sets of strategy profiles. Given a set $X \subseteq V$ of strategy profiles, we say that $X$ is \textit{eliminable} if there is a contract $\kappa$ such that $X \cap \nasheq(B^{\kappa}) = \emptyset$. Any such contract $\kappa$ is then said to \textit{eliminate} $X$.
To characterise the ability of the principal to eliminate sets of equilibria, we will utilise a graph-based approach to represent inducible deviations between strategy profiles. Given a Boolean game $B$ and an observable set $\calO$, we define the \textit{potential deviation graph} of $B$ with respect to $\calO$ to be the directed graph $\calG(B,\calO) = (\calV,E)$, where
\begin{itemize}
    \item $\calV$ is the set of vertices, where each vertex corresponds to a strategy profile $\acv \in V$;
    \item $E = \set{(\acv,\acv') \in \calV \times \calV\ \mid\ \acv \to_i \acv' \mbox{ for some } i \in N}$ is the set of directed edges, where there is an edge from a vertex $\acv \in \calV$ to another vertex $\acv' \in \calV$ if and only if it is the case that $\acv \to_i \acv'$ for some $i \in N$.
\end{itemize}

The potential deviation graph of a game $B$ wrt.\ an observable set $\calO$ represents all possible deviations that could be induced by a contract $\kappa$. 
We say that a subgraph $D = (\calV_{D},E_{D})$ of $\calG(B,\calO)$ is the \textit{deviation graph} induced by a contract $\kappa$ iff it is the case that for every pair $\acv,\acv' \in \calV_D$, we have $(\acv,\acv') \in E_{D} \Leftrightarrow \acv' \pref_i^{\kappa} \acv$. Intuitively, a contract $\kappa$ induces a deviation graph~$D$ when only those inducible deviations that are actually made beneficial for some $i \in N$ under $\kappa$ are included as edges in $D$, along with both of their endpoints. For the purposes of eliminating undesirable equilibria, we will also need to determine, given a deviation graph~$D$ of $\calG(B,\calO)$, whether there exists a contract $\kappa$ such that $D$ is the deviation graph induced by $\kappa$. If such a contract exists, we will say that $D$ is an \textit{inducible deviation graph} of $B$ wrt.\ $\calO$.

We say that a tuple $P = (\acv^1, \acv^2,\ldots, \acv^k)$ is a \textit{deviation path} if it holds that $\acv^1 \to_{i^1} \acv^2 \to_{i^2} \ldots \to_{i^{k-1}} \acv^k$, for some tuple of agents $(i^1,i^2,\ldots,i^{k-1})$. Equivalently, a deviation path may be defined simply as any directed path within the potential deviation graph $\calG(B,\calO)$. A deviation path $(\acv^1, \acv^2,\ldots \acv^k)$ is a \textit{deviation cycle} if $k \geq 2$ and $\acv^1 = \acv^k$. Additionally, we will say that a deviation path $(\acv^1, \acv^2,\ldots, \acv^k)$ \textit{involves agent $i$} if and only if we have $\acv^j \to_i \acv^{j+1}$ for some~$j \in \set{1,\ldots,k-1}$ and $i\in N$. 

Although deviation paths are defined with respect to sequences of strategy profiles that are pairwise related by inducible deviations, it will be remembered that the influence of the principal is restricted to providing incentives only based on what they can observe, rather than the strategy profile that is actually chosen. In particular, given two strategy profiles $\acv, \acv' \in V$ such that $\acv \obsneq \acv'$ and $\acv \to_i \acv'$ for some $i \in N$, the principal cannot design a contract that rewards $i$ for choosing $\acv'$ over $\acv$ directly, but this must instead be done indirectly, by rewarding $i$ sufficiently more under the \textit{observation} $\acv'\resobs$ than the observation $\acv\resobs$. From the point of view of the principal, then, it will sometimes be more useful to reason about 
\textit{observed deviation paths}, which we define as a tuple $P_{\calO} = (\obs^1,\obs^2,\ldots,\obs^k)$ such that there exists a set $ \set{\acv^1,\acv^2,\ldots,\acv^k} \subseteq \calV$ where for all $j \in \set{1,\ldots,k}$, we have 1) $\acv^j\resobs = \obs^j$ and 2) if $j < k$, then $(\acv^j,\acv^{j+1'}) \in E_D$ for some $\acv^{j+1'} \obseq \acv{j+1}$. In other words, an observed deviation path need not be an actual deviation path in $\calG(B,\calO)$; it simply needs to look like one from the principal's perspective. Given this, a deviation path $P = (\acv^1, \acv^2,\ldots, \acv^k)$ is said to contain an \textit{observed deviation cycle} if $\acv^1 \obseq \acv^j$, for some $j \in \set{2,\ldots,k}$. Terminology regarding the involvement of agents in observed deviation paths/cycles is extended in the natural way -- an observed deviation path $P_{\calO} = (\obs^1,\obs^2,\ldots,\obs^k)$ \textit{involves agent} $i$ iff $\acv^{j} \to_i \acv^{j+1}$ for some $j \in \set{1,\ldots,k-1}$ such that $\acv^j\resobs = \obs^j$ and $\acv^{j+1}\resobs = \obs^{j+1}$. 


\begin{figure}
    \centering
    \begin{tabular}{| c | c | c | c |}
        \hline
        Valuation $\acv$ & Observation $\obs$ & $\acv \in \nasheq(B)$? & $\acv \models \phi$? \hfill\\
        \hline
        (0,0,0) & 0 & \xmark & \xmark\\ 
        (0,0,1) & 0 & \xmark & \xmark\\ 
        (0,1,0) & 0 & \cmark & \xmark\\ 
        (0,1,1) & 0 & \xmark & \xmark\\ 
        (1,0,0) & 1 & \cmark & \xmark\\ 
        (1,0,1) & 1 & \xmark & \xmark\\ 
        (1,1,0) & 1 & \cmark & \xmark\\ 
        (1,1,1) & 1 & \cmark & \cmark\\ 
        \hline
    \end{tabular}\\
    \vspace{5pt}
    \includegraphics[scale=0.45]{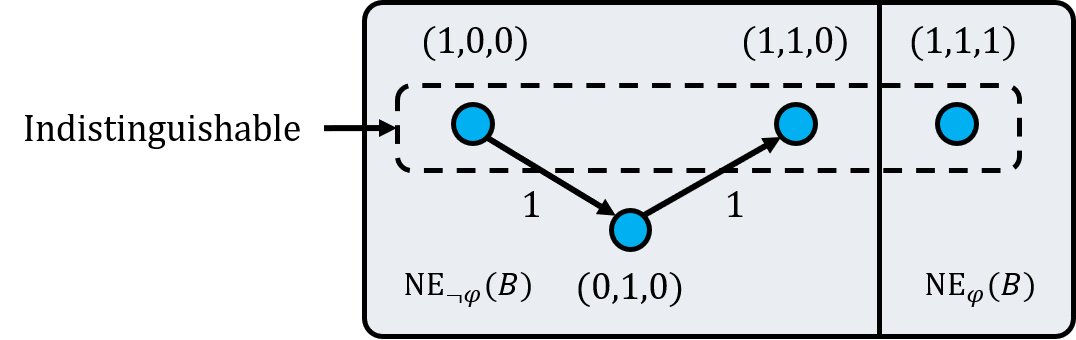}
    \caption{Example of a two-player costless game in which $\Phi = \set{p_1,p_2,p_3}$, $\Phi_1 = \set{p_1,p_2}$, $\Phi_2 = \set{p_3}$, $\gamma_1 = p_1 \vee p_2$, $\gamma_2 = p_3 \Rightarrow (p_1 \wedge p_2)$, $\calO = \set{p_1}$, $\obs_0 = 0$, $\obs_1 = (1)$ and $\phi = p_1 \wedge p_2 \wedge p_3$. The table on the top illustrates different properties of each strategy profile, and the diagram on the bottom depicts a deviation path which contains an observed deviation cycle.}
    \label{fig:deviation_cycle}
\end{figure}

\begin{example}
To illustrate some of the recently introduced concepts, consider the two-player game in Figure \ref{fig:deviation_cycle}. The diagram at the bottom depicts a deviation path $P = ((1,0,0),(0,1,0),(1,1,0))$ consisting of three Nash equilibria of $B$. Note that $P$ involves only one agent and is not a deviation cycle. Nevertheless, $P$ contains an \textit{observed} deviation cycle that involves only player $1$, as corresponding observed deviation path is $P_{\calO} = (\obs_1 ,\obs_0 ,\obs_1)$. We will see shortly that this property is crucial in determining whether a set of strategy profiles can be eliminated or not.
\end{example}

Now, we present a characterisation of the eliminability of sets of initial equilibria in a game:

\begin{theorem}
Let $B$ be a Boolean game, $\calO$ an observable set, and $X \subseteq \init{B}$ a non-empty set of initial equilibria. Then, $X$ is eliminable if and only if there exists a deviation graph $D = (\calV_D, E_D)$ of $B$ with respect to $\calO$ that satisfies the following properties: 1) For every $\acv \in X$, there is some $\acv' \in \calV_D$ such that $(\acv,\acv') \in E_D$; and 2) Every observed deviation cycle in $D$ involves at least two agents;
\end{theorem}

\begin{proof}[Proof Sketch.]
The forward direction can be shown via proof by contrapositive, which proceeds under a case analysis. If property 1) does not hold for a deviation graph $D$, then there is some strategy profile in $X$ which will not be eliminated by inducing $D$. If property 2) does not hold, then it is impossible to induce $D$, because doing so would lead to a contradiction by transitivity of the preference relation. The backward direction can be shown by constructing a contract $\kappa$ that assigns rewards using a `backward induction' procedure, which induces the deviation graph $D$. It can then be easily checked using the properties of $D$ that $\kappa$ eliminates $X$.
\end{proof}

\subsection*{Contract Design for Logical Objectives}

Returning to the problem of designing contracts for the guaranteed satisfaction of logical objectives, the natural question to ask is, given a Boolean game $B$, does there exist a contract $\kappa$ such that the principal can ensure that their goal~$\phi$ has been satisfied on some or every Nash equilibrium of the game $B^{\kappa}$? The E-Nash version of this problem is defined as follows:

\begin{quote}
\underline{\textsc{E-Nash Contractibility}}:\\
\emph{Given}: Game $B$, observable set $\calO$, formula $\phi$.\\
\emph{Question}: Does there exist a contract $\kappa$ such that for some $\obs \in \obsset$ we have $\env{B^{\kappa},\obs,\phi} \neq \emptyset$?
\end{quote}

Note that in the problem of contract design, the principal is not given an observation to begin with, but must consider \emph{all possible observations} that are consistent with at least one Nash equilibrium of a given Boolean game
-- since the agents' utilities are affected by the imposed contract, their strategies will likewise be influenced by the incentives introduced by the contract and hence, this component must be specified \emph{before} agents in the game select their strategies. We have found that this problem is no harder than the special case where the principal has full observability within the game~\cite{wooldridge2013incentive}:

\begin{theorem}\label{thm:e-nash-contractibility}
\textsc{E-Nash Contractibility} is $\SIGMA2$-complete.
\end{theorem}

\begin{proof}
For membership in $\SIGMA2$, first observe that the answer to E-Nash Contractibility is ``yes'' if and only if it is the case that $\acv \in \ind{B,\calO} \cap \nasheq_{\phi}(B) \neq \emptyset$, that is, if there exists an inducible equilibrium that satisfies $\phi$. Making use of the characterisation of inducible equilibria in Proposition \ref{proposition:inducible_eqm}, we can guess a strategy profile $\acv$ and then check whether 1) $\acv \in \nasheq_{\phi}(B)$; and 2) For all agents $i \in N$, and all choices $v_i' \in V_i$ such that $(\acv_{-i},v_i') \obseq \acv$ and $i \in W(\acv) \Leftrightarrow i \in W(\acv_{-i},v_i')$, we have that $c_i(\acv_{-i},v_i') \geq c_i(\acv)$. Checking the first condition is a coNP problem \cite{wooldridge2013incentive}. Moreover, checking the second condition is also in coNP -- simply guess an alternative choice $v_i' \in V_i$ for each agent $i \in N$ and then checking whether $(\acv_{-i},v_i') \obseq \acv$, $\acv \leftrightarrows (\acv_{-i},v_i')$, and $c_i(\acv_{-i},v_i') < c_i(\acv)$ can be done in polynomial time. If the answer is ``yes'', then condition 2 is not satisfied. Thus, the overall procedure is in the $\SIGMA2$ complexity class. 

Hardness follows again from the fact that the problem of checking whether or not a Nash equilibrium exists in a cost-free Boolean game is $\SIGMA2$-complete \cite{wooldridge2013incentive}. Given an instance of such a problem, we can check for E-Nash contractibility of the formula $\phi = \top$ in the same game, where the principal is able to observe all of the variables (i.e., $\calO = \Phi$). The answer to the E-Nash contractibility problem will be ``yes'' if and only if a Nash equilibrium exists in the corresponding cost-free Boolean game, by Proposition \ref{proposition:set_relations}.
\end{proof}

Finally, we introduce and settle the complexity of the universal counterpart to the E-Nash Contractibility problem:

\begin{quote}
\underline{\textsc{A-Nash Contractibility}}:\\
\emph{Given}: Game $B$, observable set $\calO$, formula $\phi$.\\
\emph{Question}: Is there a contract $\kappa$ s.t.\ $\nasheq(B^\kappa) \neq \emptyset$ and for all observations $\obs \in \obsset$, if $\consistent{B^{\kappa}} \neq \emptyset $, then $\consistent{B^{\kappa}} \subseteq \nasheq_{\phi}(B^{\kappa})$?
\end{quote}
With this definition in place, we can then show the following result, which sharply contrasts with results appearing in~\cite{wooldridge2013incentive}:\footnote{This contradicts Proposition~14 in \cite{wooldridge2013incentive}. Private communication with the authors has confirmed that there was an error in the original publication~\cite{wooldridge2013incentive}.} 

\begin{theorem}\label{thm:a-nash-contractibility}
\textsc{A-Nash Contractibility} is $\SIGMA3$-complete. 
\end{theorem}

\begin{proof}[Proof Sketch.]
For membership, note that the A-Nash Contractibility problem is equivalent to deciding the following statement: 
%
%
there are $\kappa \in \calK$ and $\acv \in V$ such that for all $\acv' \in V$ we have $ (\acv \in \nasheq(B^\kappa)) \wedge (\acv' \in \nasheq(B^\kappa) \Rightarrow \acv' \models \phi) \ ,$
which is a $\SIGMA3$ predicate, since there are at most an exponential number of polynomial-sized contracts to be checked\footnote{Please see the full proof in the appendix for more details.}, and checking whether $\acv \in \nasheq(B^{\kappa})$ for some $\acv \in V$ and $\kappa \in \calK$ is coNP-complete.

For hardness, we reduce from $\mbox{QSAT}_3$, which is known to be $\SIGMA3$-complete \cite{papadimitriou:94a}. Suppose that we have an instance of $\mbox{QSAT}_3$, which is given by a Boolean formula $\phi$ over a set $X$ of Boolean variables and a partition of $X$ into three sets $X_1, X_2, X_3$. The question is to decide whether $\exists X_1\forall X_2\exists X_3\ \phi$. We transform this into an instance of A-Nash Contractibility by defining a three-agent game $B = (\set{1,2,3},\Phi,(\Phi_i)_{i \in N},(\gamma_i)_{i\in N},(c_i)_{i \in N})$, where: $\Phi = X \cup \set{p,q}$, $\Phi_1 = X_1, \Phi_2 = X_2 \cup \set{p},$ and $\Phi_3 = X_3 \cup \set{q}$; $\gamma_1 = \top$; $\gamma_2 = \neg \phi \vee (p \leftrightarrow q)$; $\gamma_3 = \phi \vee \neg(p \leftrightarrow q)$; For all $\acv \in V$ and $i \in \set{1,2}$, we have $c_i(\acv) = 0$; For all $\acv \in V$, we have $c_3(\acv) = 1$ if $\acv \nvDash \phi$ and is $0$ otherwise.
The objective that the principal wishes to implement is simply $\phi$, and the observable set is given by $\calO = X_1$. The proof is completed by verifying that the answer to an instance of $\mbox{QSAT}_3$ is ``yes'' if and only if the answer to A-Nash Contractibility in the above construction is also ``yes.''
\end{proof}

This result formally establishes the intuitive notion that the task of designing incentives to eliminate \textit{all} undesirable behaviours in a multi-agent system is, in general, significantly more challenging than creating space for a desirable outcome to be chosen among many possible outcomes in the game.


\section{Concluding Remarks}
Boolean games provide a very natural model in which to draw a connection between two previously separate yet closely related areas of study: moral hazard problems (from Economics) and rational verification (from AI verification). By extending the moral hazard problem to a qualitative setting through the use of Boolean variables and propositional logic goals, this framework provides a method for expressing relationships between discrete tasks which may require a threshold of resources (costs) to complete, rather than being tied to a continuous level of effort. Our work develops this connection with a number of contributions: the formulation of a model of multi-agent moral hazard problems with combined qualitative and quantitative preferences, a characterisation of when equilibria can be induced or eliminated, and results settling the computational complexity of the verifiability and contractibility problems associated with these model of games. 

As this article and previous work on moral hazard problems demonstrate, the presence of hidden actions limits a principal's ability to design contracts that successfully align agents' local incentives with the principal's higher-level objectives. This study presents a model which opens up many possible extensions and questions. Further work may refine the model in several ways by introducing an explicit \textit{quantitative utility function for the principal} which is decreasing in contract payments, adding \textit{randomness} to observations, allowing agents to \textit{report} some of their actions, requiring that contracts be \textit{individually rational}, and letting the principal \textit{decide which observations} to make under the assumption that observations are costly. Such extensions would provide further insights into the computational aspects of incentive design in moral hazard settings, which are important concerns in AI research.


\bibliographystyle{named}
\bibliography{ijcai23}

\clearpage

\setcounter{proposition}{0}
\setcounter{theorem}{0}
\setcounter{lemma}{0}
\setcounter{corollary}{0}

\section{Supplementary Material}

\begin{proposition}\label{proposition:nash-verifiability}
\textsc{E-Nash Verifiability} is $\SIGMA2$-complete, and \textsc{A-Nash Verifiability} is $\PI2$-complete.
\end{proposition}
\begin{proof}
We consider \textsc{E-Nash Verifiability} first. For membership in $\SIGMA2$, the problem of verifying whether a candidate strategy profile $\acv$ is a Nash equilibrium of $B$ is coNP-complete \cite{wooldridge2013incentive}. Checking whether $\acv$ is consistent with $\obs$ and satisfies $\phi$ can be done in polynomial time, and thus, membership in $\SIGMA2$ follows.
Hardness follows from the fact that the problem of checking whether a Boolean game with no costs has a pure-strategy Nash equilibrium is $\SIGMA2$-complete \cite{bonzon2006boolean} -- simply set $\obs = \emptyset$, $\obs = \varepsilon$, and $\phi = \top$ in a Boolean game with no costs but the same structure as the one in which we are interested in determining the existence of Nash equilibria.

Now we consider \textsc{A-Nash Verifiability}. For membership, simply note that the A-Nash Verifiability problem is the complement of the E-Nash Verifiability problem, where the set $\env{B, \obs, \neg \phi} \neq \emptyset$ if and only if the answer to A-Nash Verifiability is ``no''.
Similarly, hardness follows by reducing the E-Nash Verifiability problem to the complement of A-Nash verifiability.
\end{proof}

\setcounter{proposition}{2}

\begin{proposition}\label{proposition:induce_dev}
Let $B$ be a game, $\calO \subseteq \Phi$ an observable set, and $\acv,\acv' \in V$ be two distinct strategy profiles such that $\acv' = (\acv_{-i},v_i')$ for some $i \in N$ and $v_i' \in V_i \setminus \set{v_i}$. Then, $\acv \to_i \acv'$ if and only if $\acv \rightharpoonup_i \acv'$ and one of the following conditions hold: 1) $\acv \obsneq \acv'$, or 2) $\acv \obseq \acv'$ and $\acv' \pref_i \acv$.
\end{proposition}

\begin{proof}
For the forward direction, suppose that it is not the case that $\acv \rightharpoonup_i \acv'$. By definition, this means that $i \in W(\acv) \cap L(\acv')$. Since any winner's utility is guaranteed to be positive and any loser's utility is guaranteed to be negative regardless of the contract chosen by the principal, there is no contract $\kappa$ such that $\acv' \pref_i^{\kappa} \acv$. Now suppose that  $\acv \rightharpoonup_i \acv'$, but that neither of two conditions hold, i.e., $\acv \obseq \acv'$ and $\acv \prefeq_i \acv'$. For any contract $\kappa$, the difference in utility $\Delta_i = u_i^{\kappa}(\acv) - u_i^{\kappa}(\acv_{-i},v_i')$ of agent $i$ between $\acv$ and $\acv'$ is
\[
 \Delta_i = \left\{\begin{array}{ll}
c_i(\acv') - c_i(\acv) & i \in W(\acv_{-i},v_i') \cap W(\acv)\\
c_i(\acv') - (1 + c_i^* +\\ \qquad \qquad \kappa_i^* + c_i(\acv)) & i \in W(\acv_{-i},v_i') \cap L(\acv)\\
c_i(\acv') - c_i(\acv) & i \in L(\acv_{-i},v_i') \cap L(\acv).\\
\end{array}\right.
\]
Two observations are in order. Firstly, because the two strategy profiles are indistinguishable, any contract will pay every agent the same amount under both strategy profiles, so the difference in utility between the two is only dependent on the initial cost structure of the game and potentially the highest possible contract payoff $\kappa_i^*$. Secondly, we can infer from the assumption that $\acv \prefeq_i \acv'$ that this difference in utility is non-negative in all three cases. Thus, we have that $\acv \prefeq_i^{\kappa} \acv'$ for all contracts $\kappa$.

For the reverse direction suppose that $\acv \rightharpoonup_i \acv'$ and that the first condition holds. Then, it can be easily verified that any contract $\kappa$ such that $\kappa_i(\acv\resobs) = 0$ and $\kappa_i(\acv'\resobs) = c_i^* + 1$ will satisfy $\acv' \pref_i^{\kappa} \acv$. Now assume instead that $\acv \obseq \acv'$ and $\acv' \pref_i \acv$. Clearly, a null contract $\kappa^0$ that simply offers all agents no payment for all observations satisfies $\acv' \pref_i^{\kappa} \acv$.
\end{proof}

\begin{proposition}\label{proposition:inducible_eqm}
Let $B$ be a game, $\calO \subseteq \Phi$ an observable set, and $\acv$ a strategy profile. Then, $\acv \in \ind{B,\calO}$ if and only if the following conditions hold: 1) $\acv \in \init{B}$; and 2) For all agents $i \in N$, and all choices $v_i' \in V_i$ such that $(\acv_{-i},v_i') \obseq \acv$ and $i \in W(\acv) \Leftrightarrow i \in W(\acv_{-i},v_i')$, we have that $c_i(\acv_{-i},v_i') \geq c_i(\acv)$.
\end{proposition}

\begin{proof}
For the forward direction, suppose that $\acv \notin \textsc{init}(B)$. Then, there is some agent $i \in L(\acv)$ and an alternative choice $v_i' \in V_i$ such that $i \in W(\acv_{-i},v_i')$. Thus, under any contract $\kappa$, we will have $u_i^{\kappa}(\acv) \leq 0$ and $u_i^{\kappa}(\acv_{-i},v_i')  \geq 1$, so $\acv \notin \nasheq(B^{\kappa})$. Since this holds for all contracts, we can conclude that $\acv$ is not an inducible equilibrium of $B$.

Now assume that the first condition is satisfied, but for some agent $i \in N$ and some $v_i' \in V_i$ such that $(\acv_{-i},v_i') \obseq \acv$ and $i \in W(\acv) \Leftrightarrow i \in W(\acv_{-i},v_i')$, it holds that $c_i(\acv_{-i},v_i') < c_i(\acv)$. First observe that under the first condition, it must be the case that $i \in W(\acv_{-i},v_i') \Rightarrow i \in W(\acv)$, otherwise $i$ would have a beneficial deviation from $\acv$ to $(\acv_{-i},v_i')$. Now, because $(\acv_{-i},v_i') \obseq \acv$, any contract $\kappa$ will pay the same amount to $i$ under $\acv$ and $(\acv_{-i},v_i')$. Thus, taking the difference in utility of agent $i$ under any contract $\kappa$ and using the assumption that $i \in W(\acv) \Leftrightarrow i \in W(\acv_{-i},v_i')$, we have $u_i^{\kappa}(\acv) - u_i^{\kappa}(\acv_{-i},v_i') = c_i(\acv_{-i},v_i') - c_i(\acv)$. Under the other assumption that $c_i(\acv_{-i},v_i') < c_i(\acv)$, it follows that $u_i^{\kappa}(\acv) < u_i^{\kappa}(\acv_{-i},v_i')$ and hence, $i$ has a beneficial deviation from $\acv$ to $(\acv_{-i},v_i')$. Since this holds for any contract, we can conclude that $\acv$ is not an inducible equilibrium of $B$.

For the reverse direction, assume that the two conditions hold. Then, consider the contract $\kappa$ such that for all agents $i \in N$ and observations $\obs \in \obsset$, we have: 
\[
\kappa_i(\obs) = \left\{\begin{array}{ll}
c_i^* + 1  & \mbox{if $\acv\resobs = \obs$}\\
0 & \mbox{otherwise.}
\end{array}\right.
\]
What remains is to show that $\acv \in \nasheq(B^{\kappa})$. Recall that because $\acv \in \init{B}$ by assumption, we have $i \in W(\acv_{-i},v_i') \Rightarrow i \in W(\acv)$, for all agents $i \in N$ and all strategies $v_i' \in V_i$. Thus, the only reason that some agent could have a beneficial deviation from $v_i$ to some $v_i'$ is if they achieve a higher net payoff, being the difference between their reward under the contract $\kappa_i$ and the cost $c_i$ that they incur. We can consider two classes of possible cases for deviation $v_i' \in V_i \setminus \set{v_i}$ by any given agent $i \in N$: cases where $(\acv_{-i},v_i') \obsneq \acv$ and cases where $(\acv_{-i},v_i') \obseq \acv$. In the first case, we know that $\kappa_i(\acv\resobs) = c_i^* + 1$ and $\kappa_i(\acv'\resobs) = 0$, so the difference in utility $\Delta_i = u_i^{\kappa}(\acv) - u_i^{\kappa}(\acv_{-i},v_i')$ for agent $i$ is
\[\resizebox{0.92\width}{0.92\height}{$
\Delta_i = \left\{\begin{array}{ll}
1 + c_i^* - c_i(\acv) + c_i(\acv_{-i},v_i') & i \in W(\acv_{-i},v_i') \cap W(\acv)\\
3 + 3c_i^* - c_i(\acv) + c_i(\acv_{-i},v_i') & i \in L(\acv_{-i},v_i') \cap W(\acv)\\
1+ c_i^* - c_i(\acv) + c_i(\acv_{-i},v_i') & i \in L(\acv_{-i},v_i') \cap L(\acv).\\
\end{array}\right.
$}\]
Note that the case where $i \in W(\acv_{-i},v_i') \cap L(\acv)$ is not possible because $\acv \in \init{B}$ by assumption. Furthermore, in all of these possibilities, the difference in utility is non-negative and hence, no agent $i \in N$ has a beneficial deviation from $\acv$ to a distinguishable strategy profile $(\acv_{-i},v_i')$. Now, for the case where $(\acv_{-i},v_i') \obseq \acv$, note that the contract payoff for both strategy profiles is the same, giving a difference in utility of
\[\resizebox{0.92\width}{0.92\height}{$
\Delta_i = \left\{\begin{array}{ll}
c_i(\acv_{-i},v_i') - c_i(\acv) & i \in W(\acv_{-i},v_i') \cap W(\acv)\\
2 + 2c_i^* - c_i(\acv) + c_i(\acv_{-i},v_i') & i \in L(\acv_{-i},v_i') \cap W(\acv)\\
 c_i(\acv_{-i},v_i') - c_i(\acv) & i \in L(\acv_{-i},v_i') \cap L(\acv).\\
\end{array}\right.$
}\]
Under the second condition, all three terms are strictly non-negative and hence, there is no beneficial deviation for any agent $i \in N$ from $\acv$ under $\kappa$. Thus, $\acv \in \ind{B,\calO}$.
\end{proof}

\begin{proposition}\label{proposition:soft}
Let $B$ be a game, $\calO \subseteq \Phi$ an observable set, and $\acv$ a strategy profile. Then, $\acv \in \soft{B,\calO}$ if and only if the following conditions hold:
\begin{itemize}
    \item $\acv \in \init{B}$, and
    \item There is an agent $i \in N$ and a strategy $v_i' \in V_i \setminus \set{v_i}$ such that $(\acv_{-i},v_i') \obsneq \acv$ and $i \in W(\acv_{-i},v_i') \Leftrightarrow i \in W(\acv)$.
\end{itemize}
\end{proposition}

\begin{proof}
For the forward direction, it is immediate that if $\acv \notin \init{B}$, then $\acv \notin \soft{B,\calO} \subseteq \init{B}$. Hence, assume that $\acv \in \init{B}$ and suppose that for all agents $i \in N$ and all strategies $v_i \in V_i \setminus \set{v_i}$ such that $(\acv_{-i},v_i') \obsneq \acv$, it is not the case that $i \in W(\acv_{-i},v_i') \Leftrightarrow i \in W(\acv)$. Now, for a given agent $i \in N$, it is not possible that $i \in L(\acv) \cap W(\acv_{-i},v_i')$ because otherwise, $\acv$ would not be an initial equilibrium. Thus, it must be that $i \in W(\acv) \cap L(\acv_{-i},v_i')$, so no contract $\kappa$ could convince any agent $i$ to deviate to some distinguishable strategy profile $(\acv_{-i},v_i') \obsneq \acv$. Also, because $\acv \in \init{B}$, no agent has a beneficial deviation to any indistinguishable strategy profile. Moreover, because no contract can induce different costs for indistinguishable strategy profiles, there is no contract that can induce a beneficial deviation from $\acv$ for any agent. Thus, we have that $\acv \notin \soft{B,\calO}$.

For the reverse direction, assume that both conditions hold. Then, it is easy to check that any contract $\kappa$ such that $\kappa_i(\acv_{-i},v_i') = c_i^* + 1$ and $\kappa_i(\acv) = 0$ induces a beneficial deviation for $i$ from $\acv$ to $(\acv_{-i},v_i')$. Thus, it follows that $\acv \in \soft{B,\calO}$.
\end{proof}

\begin{theorem}\label{thm:eliminable}
Let $B$ be a Boolean game, $\calO$ an observable set, and $X \subseteq \init{B}$ a non-empty set of initial equilibria. Then, $X$ is eliminable if and only if there exists a deviation graph $D = (\calV_D, E_D)$ of $B$ with respect to $\calO$ that satisfies the following properties: 1) For every $\acv \in X$, there is some $\acv' \in \calV_D$ such that $(\acv,\acv') \in E_D$; and 2) Every observed deviation cycle in $D$ involves at least two agents;
\end{theorem}

\begin{proof}
For the forward direction, suppose that for all deviation graphs, one of the two properties does not hold. Let $D$ be any deviation graph of $B$ and suppose that property $(1)$ does not hold. Then, there is some $\acv \in X$ for which there is no $\acv' \in \calV_D$ such that $(\acv,\acv') \in E_D$. Therefore, inducing the deviation graph $D$ will not eliminate $X$. Next, suppose that there is some observed deviation cycle $P = (\acv^1,\acv^{2},\ldots,\acv^k)$ in $D$ that involves only one agent $i \in N$. Then, any contract $\kappa$ that induces the cycle $P$ must be such that $\acv^k \pref_i^{\kappa} \acv^{k-1} \pref_i^{\kappa} \ldots \pref_i^{\kappa} \acv^1$. By transitivity of the preference relation, it must be the case that $\acv^k \pref_i^{\kappa} \acv^1 $. By definition of an observed deviation cycle, we also have that $\acv^k \obseq \acv^1$. Moreover, since both $\acv^1$ and $\acv^k$ are initial equilibria of $B$, it must be the case that $\acv^1 \leftrightarrows_i \acv^k$ and $c_i(\acv^1) = c_i(\acv^k)$, otherwise $i$ would have a beneficial deviation from one to the other, even in the absence of contracts. However, this does not satisfy the conditions of Proposition \ref{proposition:induce_dev} and thus, no contract exists such that $\acv^k \pref_i^{\kappa} \acv^1$, giving rise to a contradiction. Hence, we can conclude that for every deviation graph $D$, either $D$ is not inducible by any contract $\kappa$, or inducing $D$ will not eliminate $X$. Therefore, there is no contract $\kappa$ that eliminates $X$.

For the backward direction, suppose that there exists a deviation graph $D$ which satisfies both of the given properties.
Then, let $\indev{\obs}$ denote the set of agents $i \in N$ for which there is some $\acv \in \consistent{B}$ and $\acv' \in \calV_D$ such that $(\acv',\acv) \in E_D$.
For each $i \in N$, let $\ell_i$ be the length of the longest observed deviation path in $D$ that involves only $i$. Additionally, for all $\obs \in \obsset$, let $d_i(\obs)$ denote the length of the longest observed deviation path in $D$ \textit{starting from} $\obs$ that involves only $i$. Note that because no observed deviation cycle involves only one agent by assumption, $\ell_i$ is finite for all $i \in N$ and $\obs \in \obsset$. Thus, we can construct a contract $\kappa$ as follows:

\[
\kappa_i(\obs) = \left\{\begin{array}{ll}
(\ell_i - d_i(\obs))(c_i^* + 1)  & \mbox{if $i \in \indev{\obs}$}\\
0 & \mbox{otherwise.}
\end{array}\right.
\]

Intuitively, $\kappa$ is designed so that for every edge $(\acv,\acv') \in E_D$, the agent who has an inducible deviation from $\acv$ to $\acv'$ is offered a reward when $\acv'\resobs$ is observed that is significantly greater than when $\acv\resobs$ is observed. More precisely, observe first that for all $(\acv,\acv') \in E_D$, we have $d_i(\acv\resobs) > d_i(\acv'\resobs)$ for the agent $i$ such that $\acv \to_i \acv'$, because $D$ is assumed not to contain any observed deviation cycles involving only one agent. Thus, $\kappa$ offers $i$ at least $c_i^*+1$ more under $\acv'\resobs$ than under $\acv\resobs$. By the definition of inducible deviations, we can then conclude that for any $(\acv,\acv') \in E_D$, we have $\acv' \pref_i^{\kappa} \acv$. Then, by the first property, it holds that for every $\acv \in X$, there is some $\acv' \in V$ such that $\acv' \pref_i^{\kappa} \acv$ for some $i \in N$, so $\kappa$ eliminates $X$.
\end{proof}

\begin{theorem}\label{thm:e-nash-contractibility}
\textsc{E-Nash Contractibility} is $\SIGMA2$-complete.
\end{theorem}

\begin{proof}
For membership in $\SIGMA2$, first observe that the answer to E-Nash Contractibility is ``yes'' if and only if it is the case that $\acv \in \ind{B,\calO} \cap \nasheq_{\phi}(B) \neq \emptyset$, that is, if there exists an inducible equilibrium that satisfies $\phi$. Making use of the characterisation of inducible equilibria in Proposition \ref{proposition:inducible_eqm}, we can guess a strategy profile $\acv$ and then check whether 1) $\acv \in \nasheq_{\phi}(B)$; and 2) For all agents $i \in N$, and all choices $v_i' \in V_i$ such that $(\acv_{-i},v_i') \obseq \acv$ and $i \in W(\acv) \Leftrightarrow i \in W(\acv_{-i},v_i')$, we have that $c_i(\acv_{-i},v_i') \geq c_i(\acv)$. Checking the first condition is a coNP problem \cite{wooldridge2013incentive}. Moreover, checking the second condition is also in coNP -- simply guess an alternative choice $v_i' \in V_i$ for each agent $i \in N$ and then checking whether $(\acv_{-i},v_i') \obseq \acv$, $\acv \leftrightarrows (\acv_{-i},v_i')$, and $c_i(\acv_{-i},v_i') < c_i(\acv)$ can be done in polynomial time. If the answer is ``yes'', then condition 2 is not satisfied. Thus, the overall procedure is in the $\SIGMA2$ complexity class. 

Hardness follows again from the fact that the problem of checking whether or not a Nash equilibrium exists in a cost-free Boolean game is $\SIGMA2$-complete \cite{wooldridge2013incentive}. Given an instance of such a problem, we can check for E-Nash contractibility of the formula $\phi = \top$ in the same game, where the principal is able to observe all of the variables (i.e., $\calO = \Phi$). The answer to the E-Nash contractibility problem will be ``yes'' if and only if a Nash equilibrium exists in the corresponding cost-free Boolean game, by Proposition \ref{proposition:set_relations}.
\end{proof}

\begin{theorem}\label{thm:a-nash-contractibility}
\textsc{A-Nash Contractibility} is $\SIGMA3$-complete. 
\end{theorem}

\begin{proof}
For membership, note that the A-Nash Contractibility problem is equivalent to deciding the following statement: 
\begin{flushleft}
$\qquad \exists \kappa \in \calK, \exists \acv \in V, \forall \acv' \in V,$
\end{flushleft}
\begin{flushright} $ (\acv \in \nasheq(B^\kappa)) \wedge (\acv' \in \nasheq(B^\kappa) \Rightarrow \acv' \models \phi) \ ,$ \end{flushright}
which is a $\SIGMA3$ predicate. This follows from two main observations: firstly, the contract $\kappa$ constructed in the backward direction in the proof of Theorem \ref{thm:eliminable} is guaranteed to eliminate 
$\nasheq_{\neg \phi}(B^{\kappa})$ if $\phi$ is A-Nash contractible. The maximum payment awarded by such a contract is $\ell_i(c_i^* + 1)$, where we recall that $\ell_i$ is the length of the longest observed deviation path in $D$ that involves only $i$. The length of such a path is bounded from above by the longest non-cyclic path through the set of all strategy profiles, which is $2^{\vert \Phi \vert}$. Thus, if $\phi$ is A-Nash contractible, then a contract can be designed to achieve this using payments of at most $2^{\vert \Phi\vert} \cdot \max_{i\in N}(c_i^* + 1)$, which can be represented using a polynomial number of bits. Thus, the size of each contract is polynomial in the size of the game, since the size of the original cost functions $c_i$ are already exponential in $\vert \Phi \vert$. There are at most an exponential number of contracts to check with this maximum value and hence, guessing a contract can be done in NP. Secondly, the problem of checking whether $\acv \in \nasheq(B^{\kappa})$ for some $\acv \in V$ and $\kappa \in \calK$ is in fact coNP-complete. These combined give membership in $\SIGMA3$.

For hardness, we reduce from $\mbox{QSAT}_3$, which is known to be $\SIGMA3$-complete \cite{papadimitriou:94a}. Suppose that we have an instance of $\mbox{QSAT}_3$, which is given by a Boolean formula $\phi$ over a set $X$ of Boolean variables and a partition of $X$ into three sets $X_1, X_2, X_3$. The question is to decide whether $\exists X_1\forall X_2\exists X_3\ \phi$. We transform this into an instance of A-Nash Contractibility by defining a three-agent game $B = (\set{1,2,3},\Phi,(\Phi_i)_{i \in N},(\gamma_i)_{i\in N},(c_i)_{i \in N})$, where: $\Phi = X \cup \set{p,q}$, $\Phi_1 = X_1, \Phi_2 = X_2 \cup \set{p},$ and $\Phi_3 = X_3 \cup \set{q}$; $\gamma_1 = \top$; $\gamma_2 = \neg \phi \vee (p \leftrightarrow q)$; $\gamma_3 = \phi \vee \neg(p \leftrightarrow q)$; For all $\acv \in V$ and $i \in \set{1,2}$, we have $c_i(\acv) = 0$; For all $\acv \in V$, we have $c_3(\acv) = 1$ if $\acv \nvDash \phi$ and is $0$ otherwise.
The objective that the principal wishes to implement is simply $\phi$, and the observable set is given by $\calO = X_1$. We now show that the answer to an instance of $\mbox{QSAT}_3$ is ``yes'' if and only if the answer to A-Nash Contractibility in the above construction is also ``yes.''

Suppose that the answer to the instance of $\mbox{QSAT}_3$ is ``yes'' and let $v_1^* = X_1$ for some assignment $X_1$ that satisfies the $\mbox{QSAT}_3$ instance. Then, consider the contract $\kappa$ such that for all $i \in N$ and $\obs \in \calO$:
\[
\kappa_i(\obs) = \left\{\begin{array}{ll}
c_1^* + 1 & \obs = v_1^* \mbox{ and } i=1 \\
0 & \mbox{otherwise}
\end{array}\right.
\]
Intuitively, $\kappa$ offers agent $1$ a very high reward if they do choose $v_1^*$ and offers no reward to agents $2$ or $3$. Let $\acv$ be a strategy profile such that the following three conditions hold:
\begin{enumerate}
    \item $v_1 = v_1^*$;
    \item $\acv \models \phi$;
    \item $p \leftrightarrow q$.
\end{enumerate}
Such a strategy profile is guaranteed to exist under the assumption that we have a ``yes'' instance of $\mbox{QSAT}_3$. Then, it is easy to check that all agents achieve their goals and minimise their costs under $\kappa$, so there is no incentive to unilaterally deviate. Thus, $\acv \in \nasheq_\phi(B^{\kappa})$. 

Now, suppose that $\acv$ is a strategy profile that does not satisfy one of the three conditions. If $v_1 \neq v_1^*$, then agent $1$ has a beneficial deviation to $v_1^*$ to reduce their cost. Thus, assume that $v_1 = v_1^*$ and suppose that $\acv \nvDash \phi$. Then, because $v_1 = v_1^*$, no matter what strategy agent 2 chooses, agent 3 has a deviation that satisfies $\phi$ and will benefit from doing so because their cost is strictly lower when $\phi$ is satisfied than when it is not. Thus, agent 3 has a beneficial deviation to achieve their goal. Finally, if $\neg (p \leftrightarrow q)$ holds, then agent 2 has a beneficial deviation by setting $p$ to the same value as $q$. Thus, any strategy profile that does not satisfy the three conditions is not a Nash equilibrium of $B^\kappa$. This means all Nash equilibria under $B^{\kappa}$ satisfy $\phi$, and hence, the answer to A-Nash Contractibility is ``yes.''

Suppose then, that the answer to the instance of $\mbox{QSAT}_3$ is ``no.'' Then, it holds that $\forall X_1 \exists X_2 \forall X_3 \neg \phi$. Let $\kappa$ be any contract and let $\resp{v_1} = \set{v_2 \in V_2 \mid \forall v_3 \in V_3\ (v_1,v_2,v_3) \nvDash \phi}$. Then, there exists a strategy profile $\acv$ satisfying the following three conditions:
\begin{enumerate}
    \item $\kappa_1(\acv\resobs) - c_1(\acv) \geq \kappa_1((\acv_{-1},v_1')\resobs) - c_1(\acv_{-1},v_1')$,~~\\ \hfill for all $v_1' \in V_1$;
    \item $v_2 \in \resp{v_1}$;
    \item $\acv \models \neg (\phi \vee (p \leftrightarrow q))$.
\end{enumerate}
To see why $\acv$ necessarily exists, firstly note that because the answer to this instance of $\mbox{QSAT}_3$ is ``no'', then it follows that $\resp{v_1} \neq \emptyset$ for all $v_1 \in V_1$ so condition 2 can always be satisfied for any strategy that agent 1 chooses. Next, observe that a strategy profile which satisfies condition 1 always exists, because $V_1$ is finite. Finally, it is clear that because $v_2 \in \resp{v_1}$, we have $\acv \models \neg \phi$ and no matter what value agent 3 assigns to $q$, agent 2 has an assignment to $p$ so that $\acv \models \neg (p \leftrightarrow q)$.

Now, under $\acv$ and $\kappa$, agent $1$ achieves the lowest cost attainable by the property $1$ of $\acv$ and trivially achieves their goal, so there is no incentive for them to deviate. Moreover, both agents 2 and 3 achieve their goals by property $3$, so there is no incentive for them to deviate in order to achieve their goals. Thus, the only possible deviation that could occur from $\acv$ is one by either agents 2 or 3 in order to reduce their costs, whilst retaining goal satisfaction. However, for any fixed $v_1$ and contract $\kappa$, all choices made by agents 2 and 3 will result in the same contract payoff to them, because none of their actions are observable by the principal. Additionally, agent 3 cannot deviate in order to satisfy $\phi$ and reduce their cost because $v_2 \in \resp{v_1}$. Hence, there is no incentive for either agent to deviate in pursuit of lower costs. Since no beneficial deviation exists from $\acv$, we have $\acv \in \nasheq_{\neg \phi}(B^{\kappa})$. Thus, for all $\kappa$, we have $\nasheq_{\neg \phi}(B^{\kappa}) \neq \emptyset$, so the answer to A-Nash Contractibility is ``no.''
\end{proof}

\end{document}